\documentclass[9pt,conference]{IEEEtran}
\pdfoutput=1
\IEEEoverridecommandlockouts
\usepackage{cite}
\usepackage{amsmath,amssymb,amsfonts,amsthm}
\usepackage{graphicx}
\usepackage{subfigure}
\usepackage{textcomp}
\usepackage{xcolor}
\usepackage{algorithm}
\usepackage{algpseudocode}
\usepackage{mathrsfs}
\usepackage{setspace}
\usepackage{subfig}
\usepackage{epsfig}  
\usepackage{multirow} 
\usepackage{calc}
\usepackage{tikz}
\usetikzlibrary{patterns}
\usetikzlibrary{arrows}
\usepackage{sansmath}
\tikzset{
     task/.style={fill=#1,  rectangle},
     task2b/.style={task=orange,draw,minimum height=\uy},
     task7/.style={draw,minimum height=\uy,},
     task9/.style={task=lightgray,draw,minimum height=\uy,},
}
\tikzset{>=latex}
\tikzstyle{circleNode}=[circle,draw=blue!75,fill=blue!20,minimum
size=6mm]
\tikzstyle{niceFill}=[draw=blue!75,fill=blue!20,minimum size=6mm]
\def\uy{0.5cm} 
\usepackage{color}

\newtheorem{theorem}{Theorem}
\newtheorem{lemma}{Lemma}

\def\BibTeX{{\rm B\kern-.05em{\sc i\kern-.025em b}\kern-.08em
    T\kern-.1667em\lower.7ex\hbox{E}\kern-.125emX}}
\begin{document}

\title{DPCP-p: A Distributed Locking Protocol for Parallel Real-Time Tasks
\thanks{Work supported by the NSFC (Grant No. 61802052) and the China Postdoctoral Science Fundation Funded Project (Grant No. 2017M612947).}
}

\author{
Maolin Yang\textsuperscript{$\dagger$} \quad\quad Zewei Chen\textsuperscript{$\dagger$} \quad\quad  Xu Jiang\textsuperscript{$\dagger$} \quad\quad Nan Guan\textsuperscript{$\ddagger$} \quad\quad Hang Lei\textsuperscript{$\dagger$}
\\
\textsuperscript{$\dagger$}\emph{University of Electronic Science and Technology of China (UESTC), Chengdu, China} \\
\textsuperscript{$\ddagger$}\emph{Hong Kong Polytechnic University (PolyU), Hong Kong, China}
}

\maketitle
\thispagestyle{plain}
\pagestyle{plain}

\begin{abstract}
Real-time scheduling and locking protocols are fundamental facilities to construct time-critical systems. For parallel real-time tasks, predictable locking protocols are required when concurrent sub-jobs mutually exclusive access to shared resources. This paper for the first time studies the distributed synchronization framework of parallel real-time tasks, where both tasks and global resources are partitioned to designated processors, and requests to each global resource are conducted on the processor on which the resource is partitioned. We extend the Distributed Priority Ceiling Protocol (DPCP) for parallel tasks under federated scheduling, with which we proved that a request can be blocked by at most one lower-priority request. We develop task and resource partitioning heuristics and propose analysis techniques to safely bound the task response times. Numerical evaluation (with heavy tasks on 8-, 16-, and 32-core processors) indicates that the proposed methods improve the schedulability significantly compared to the state-of-the-art locking protocols under federated scheduling.
\end{abstract}


\section{Introduction}
To exploit the parallelism for time-critical applications on multicores, the design of scheduling and analysis techniques for parallel real-time tasks has attracted increasing interests in recent years. Among the scheduling algorithms for parallel real-time tasks, the federated scheduling~\cite{DBLP:conf/ecrts/LiCALGS14} is a promising approach with high flexibility and simplicity in analysis.   

Coordinated locking protocols are used to ensure mutually exclusive access to shared resources while preventing uncontrolled priority inversions~\cite{DBLP:journals/tpds/DinhLAGL18,DBLP:conf/dac/JiangGLY19}. In multiprocessor systems, requests to shared resources can be executed locally by the tasks~\cite{DBLP:conf/icdcs/Rajkumar90} or remotely by resource agents~\cite{DBLP:conf/rtss/RajkumarSL88}, e.g., by means of the Remote Procedure Call (RPC) mechanism. Local execution of requests is in general more efficient since migrations are not needed, while blockings can be better explored and managed with remote execution of requests, e.g., by constraining resource contentions on designated processors~\cite{DBLP:conf/emsoft/HsiuLK11,DBLP:conf/rtss/HuangYC16}. While existing locking protocols for parallel tasks~\cite{DBLP:journals/tpds/DinhLAGL18,DBLP:conf/dac/JiangGLY19} are all based on local execution of requests, no work has been done with remote execution of requests so far as we know.

The Distributed Priority Ceiling Protocol (DPCP)~\cite{DBLP:conf/rtss/RajkumarSL88} is a classic multiprocessor real-time locking protocol for sequential tasks that executes requests of global resources remotely, where both tasks and shared resources are partitioned among the processors and all requests to a global resource must be conducted by the resource agents on the processor on which the resource is partitioned. Empirical studies~\cite{DBLP:conf/rtas/Brandenburg13} indicate that the DPCP has better schedulability performance compared to similar protocols with local execution of requests. Further, the recent Resource-Oriented Partitioned (ROP) scheduling~\cite{DBLP:conf/rtss/HuangYC16,DBLP:conf/rtns/BruggenCHY17,DBLP:journals/tc/YangHC19} with the DPCP guarantees bounded speedup factors.

In addition, since each heavy task exclusively uses a subset of processors under federated scheduling, there could be significant resource waste under the federated scheduling, i.e., almost half of the processing capacity is wasted in the extreme case. Executing global-resource-requests on remote processors can alleviate the potential resource wastes by shifting a part of the resource-related workload of a task to processors with lower workload. 

This paper for the first time studies the distributed synchronization framework for parallel real-time tasks. we answer the fundamental question of whether the key insight of remote execution of shared resources for sequential tasks can be applied to parallel real-time tasks and how to do so. We propose \emph{DPCP-p}, an extension of DPCP, to support parallel real-time tasks under federated scheduling, and develop the corresponding schedulability analysis and partitioning heuristic. DPCP-p retains the fundamental property of the underlying priority ceiling mechanism of the DPCP, namely a request can be blocked by at most one lower-priority request. Numerical evaluation with heavy tasks on more than 8-core processors indicates that DPCP-p improves the schedulability performance significantly compared to existing locking protocols under federated scheduling.

\iffalse
In this paper, we try to answer the fundamental question whether the key insight of remote execution of shared resources for sequential tasks can be applied to parallel real-time tasks and how to do so ? We extend the classic DPCP, called DPCP-p, to support real-time locks for parallel tasks under federated scheduling, and shows that the underlying priority ceiling mechanism of the DPCP still work under DPCP-p: a request can be locked by at most one lower-priority request. 

Since each heavy task exclusively use a subset of processors, there could be significant resource waste under the federated scheduling~\cite{DBLP:conf/rtss/JiangGL017}. %, i.e., almost half of the processing capacity is wasted in the extreme case~\cite{DBLP:conf/rtss/JiangGL017}. 
The distributed synchronization approach 

Existing locking protocols for parallel tasks~\cite{DBLP:journals/tpds/DinhLAGL18,DBLP:conf/dac/JiangGLY19} are all based on local execution of shared resources. 
\fi

\section{System Model and Terminologies}
\label{sec:model}
We consider a set of $n$ parallel tasks $\tau=\{\tau_1,...,\tau_n\}$ to be scheduled on $m\geq 2$ identical processors $\wp=\{\wp_1,...,\wp_m\}$ with $n_r$ shared resources $\Phi=\{\ell_1,...,\ell_{n_r}\}$. 

\noindent \textbf{Parallel Tasks.}
Each task $\tau_i$ is characterized by a Worst-Case Execution Time (WCET) $C_i$, a relative deadline $D_i$, and a minimum inter-arrival time $T_i$, where $D_i\leq T_i$ (\emph{constrained-deadline} is considered). The utilization of $\tau_i$ is defined by $U_i=C_i/T_i$. 

The structure of $\tau_i$ is represented by a \emph{Directed Acyclic Graph (DAG)} $G_i=\langle V_i,E_i \rangle$, where $V_i$ is the set of vertices and $E_i$ is the set of edges. Each vertex $v_{i,x}\in  V_i$ has a WCET $C_{i,x}$, and the WCET of all vertices of $\tau_i$ is $C_i=\sum_{v_{i,x}\in V_i}C_{i,x}$. Each edge $(v_{i,x},v_{i,y})\in E_i$ represents the precedence relation between $v_{i,x}$ and $v_{i,y}$. A vertex $v_{i,x}$ is said to be \emph{pending} during the time while all its predecessors are finished and $v_{i,x}$ is not finished. 
A complete path is a sequence of vertices $(v_{i,a},...,v_{i,z})$, where $v_{i,a}$ is a head vertex, $v_{i,z}$ is a tail vertex, and $v_{i,x}$ is the predecessor of $v_{i,y}$ for each pair of consecutive vertices $v_{i,x}$ and $v_{i,x+1}$. We use $\lambda_i$ to denote an arbitrary complete path. The length of $\lambda_i$, denoted by $\mathcal{L}(\lambda_i)$, is defined as the sum of the WCETs of the vertices on $\lambda_i$. We also use $\mathcal{L}_i^{\ast}$ to denote the length of the longest path of $G_i$. For example in Fig.~1(a), the longest path of $G_i$ is $(v_{i,1},v_{i,5},v_{i,7},v_{i,8})$, and $\mathcal{L}_i^{\ast}=10$  

At runtime, each task generates a sequence of jobs, and each job inherits the DAG structure of the task. Let $J_{i,j}$ denote the $j$th job of $\tau_i$. Let $a_{i,j}$ and $f_{i,j}$ denote the arrival and finish time of $J_{i,j}$ respectively, then $J_{i,j}$ must finish no later than $a_{i,j}+D_i$, and the subsequent job $J_{i,j+1}$ cannot arrive before $a_{i,j}+T_i$. The Worst-Case Response Time (WCRT) of task $\tau_i$ is defined as $R_i=\max_{\forall j}\{f_{i,j}-a_{i,j}\}$. For brevity, let $J_i$ be an arbitrary job of $\tau_i$.

\noindent \textbf{Shared Resources.}
Each task $\tau_i$ uses a set of shared resources $\Phi_i\subseteq\Phi$, and each resource $\ell_q$ is shared by a set of tasks $\tau(\ell_q)$. To ensure mutual exclusion, $\ell_q$ is protected by a \emph{binary semaphore} (also called a \emph{lock} for short). A job is allowed to execute a \emph{critical section} for $\ell_q$ only if it holds the lock of $\ell_q$, otherwise, it is suspended. A vertex $v_{i,x}$ requests $\ell_q$ at most $N_{i,x,q}$ times, and each time uses $\ell_q$ for a time of at most $L_{i,q}$. For simplicity, we assume that a path $\lambda_i$ requests $\ell_q$ at most $N^{\lambda}_{i,q}=\sum_{v_{i,x}\in\lambda_i}N_{i,x,q}$ times, and a job $J_i$ requests $\ell_q$ at most $N_{i,q}=\sum_{v_{i,x}\in V_i}N_{i,x,q}$ times.

Given that $L_{i,q}$ is included in $C_i$, for brevity, we use $C^{\prime}_{i,x}$ and $C_i^{\prime}$ to denote the WCETs of the non-critical sections of $v_{i,x}$ and $\tau_i$, respectively. For simplicity, it is assumed that $C_i^{\prime}=\sum_{v_{i,x}\in\tau_i}C_{i,x}^{\prime}=C_i-\sum_{\ell_q\in \Phi_i}N_{i,q}\cdot L_{i,q}$. Further, critical sections are assumed to be non-nested, and nested critical sections remain in future work. 

\noindent \textbf{Scheduling.}
The tasks are scheduled based on the \emph{federated scheduling} paradigm~\cite{DBLP:conf/ecrts/LiCALGS14}. Each task $\tau_i$ with $C_i/D_i>1$ (i.e., \emph{heavy tasks}) is assigned $m_i$ dedicated processors, and we use $\wp(\tau_i)$ to denote the set of processors assigned to $\tau_i$. The rest \emph{light tasks} are assigned to the remaining processors. Each task $\tau_i$ has a unique base priority $\pi_i$, and $\pi_i<\pi_h$ implies that $\tau_i$ has a base priority lower than $\tau_h$. All jobs of $\tau_i$ and all vertices of $\tau_i$ have the same base priority $\pi_i$. 

At runtime, each heavy task is scheduled exclusively on the assigned processors according to a \emph{work-conserving} scheduler (i.e., no processor assigned to a task is idle when there is a vertex of the task is ready to be scheduled), while each light task is treated as a sequential task and is scheduled with the tasks (if one exists) assigned on the same processor. 
We focus on heavy tasks in the following and discuss how to handle both heavy and light tasks in Sec.~\ref{sec:discussions}.
 
\section{The Distributed locking protocol DPCP-p} 
\label{sec:DPCP-p}
The design of DPCP-p is based on the DPCP~\cite{DBLP:conf/rtss/RajkumarSL88} and extents it to support parallel real-time tasks under federated scheduling. 

\subsection{The Synchronization Framework}
Under federated scheduling, a resource can be shared locally or globally. A resource $\ell_q$ is a \emph{local resource} if it is shared only by the vertices of a single task, and it is a \emph{global resource} if it is shared by more than one task. For example in Fig. 1, $\ell_1$ is a global resource and $\ell_2$ is a local resource. We use $\Phi^L$ and $\Phi^G$ to denote the local resources and the global resources respectively.

Each global resource $\ell_q\in \Phi^G$ is assigned to a processor, and all requests to $\ell_q$ must execute on that processor, e.g., by means of an RPC-like agent~\cite{DBLP:conf/rtss/RajkumarSL88}. Once a vertex requests a global resource, it is suspended until the agent finishes. Requests to local resources are executed by the tasks directly, i.e., no migration is required. 

For brevity, we use $\Phi(\wp_k)$ to denote the set of global resources on processor $\wp_k$. The global resources that are assigned to the same processor as $\ell_q$ are denoted by $\Phi^{\wp}(\ell_q)$, and the global resources that are assigned to the same processors as $\tau_i$ are denoted by $\Phi^{\wp}(\tau_i)$.

\subsection{Queue Structure}
While pending, a vertex is either executing, ready and not scheduled, or suspended. The following queues are used to maintain the states of the vertices for each task.
\begin{itemize}
\item $RQ^N_i$: the ready queue of $\tau_i$ for the vertices that are ready to execute non-critical sections. The vertices in $RQ^N_i$ are scheduled in a First In First Out (FIFO) order. 
\item $RQ^L_i$: the ready queue of $\tau_i$ for the vertices that are holding local resources. The vertices in $RQ^L_i$ are scheduled in a FIFO order. If both $RQ^N_i$ and $RQ^L_i$ are not empty, the vertices in $RQ^L_i$ are scheduled first .
\item $SQ_i$: the suspended queue of $\tau_i$. Each vertex in $SQ_i$ is waiting for a request to be finished.
\end{itemize}

In addition, each processor maintains two hybrid queues to handle the global-resource-requests.
\begin{itemize}
\item $RQ^G_k$: the ready queue of the global-resource-requests on processor $\wp_k$. The requests in $RQ^G_k$ are scheduled by the priorities of the tasks.
\item $SQ^G_k$: the suspended queue of the global-resource-requests on processor $\wp_k$.
\end{itemize}  

\subsection{Locking Rules}
Under priority scheduling, the problem of \emph{priority inversion}~\cite{DBLP:conf/rtss/BrandenburgA10} is inevitable when jobs compete for shared resources. Various progress mechanisms~\cite{DBLP:conf/rtss/RajkumarSL88,DBLP:conf/icdcs/Rajkumar90,DBLP:conf/rtss/BrandenburgA10} are used to minimize the duration of priority inversions. We consider the inherent priority ceiling mechanism as used in the DPCP~\cite{DBLP:conf/rtss/RajkumarSL88} in the following.

Consider a global resource $\ell_q\in\Phi^G$ on processor $\wp_k$, the \emph{priority ceiling} of $\ell_q$ is defined as $\Pi_q=\pi^H+\max_{\tau_j\in\tau(\ell_q)}\pi_j$, where $\pi^H$ is a priority level higher than the base priority of any task in $\tau$. At runtime, the \emph{processor ceiling} of $\wp_k$ at some time $t$, denoted by $\Pi^{\wp}_k(t)$, is the maximum of the priority ceilings of the global resources that are allocated to $\wp_k$ and locked at time $t$. Let $\Re_{i,q}$ be a request from a job $J_i$ to a global resource $\ell_q\in\Phi^G$. The \emph{effective priority} of $\Re_{i,q}$ at some time $t$, denoted by $\pi^E_i(t)$, is elevated to $\pi^E_i(t)=\pi^H+\pi_i$. The priority ceiling mechanism ensures that: a global-resource-request $\Re_{i,q}$ is granted the lock at time $t$ only if $\pi^E_i(t)>\Pi^{\wp}_k(t)$.

Next, we introduce the locking rules of DPCP-p. Consider a vertex $v_{i,x}$ issues a request $\Re_{i,q}$ for a resource $\ell_q$ at some time $t$.

\textbf{Rule 1.} If $\ell_q$ is a local resource locked by another vertex at time $t$, then $v_{i,x}$ is suspended and enqueued to $SQ_i$.

\textbf{Rule 2.} If $\ell_q$ is a local resource not locked at time $t$, then $v_{i,x}$ locks $\ell_q$ and queues upon $RQ^L_i$, i.e., $v_{i,x}$ is ready to be scheduled to execute the critical section. 

\textbf{Rule 3.} If $\ell_q$ is a global resource on some processor $\wp_k$, then $v_{i,x}$ is suspended and enqueued to $SQ_i$. Meanwhile, $\Re_{i,q}$ tries to lock $\ell_q$ according to the priority ceiling mechanism. $\Re_{i,q}$ queues upon $RQ^G_k$ and is ready to be scheduled (according to its priority) if the lock is granted, otherwise $\Re_{i,q}$ is enqueued to $SQ^G_k$.

\textbf{Rule 4.} Once $\Re_{i,q}$ finishes, it releases the lock of $\ell_q$, and dequeues from $RQ^G_i$ if $\ell_q$ is a global resource. Then, $v_{i,x}$ is enqueued to $RQ^N_i$.

Fig. 1 shows an example schedule of DPCP-p with two DAG tasks on a four-core processor, and each task is assigned two processors. At time $t=2$, (i) $v_{i,2}$ is suspended and enqueued to $SQ_i$ until the global-resource-request $\Re_{i,1}$ finishes on processor $\wp_2$ at time $t=7$, (ii) $\Re_{i,1}$ is suspended and enqueued to $SQ_2^G$ until $\Re_{j,1}$ releases $\ell_1$ at time $t=4$, and (iii) $v_{i,3}$ locks a local resource $\ell_2$, enqueued to $RQ^L_i$, and is scheduled until time $t=4$, while $v_{i,4}$ is suspended and queued upon $SQ_i$ until $v_{i,3}$ releases $\ell_2$ at time $t=4$.

%Consider two parallel tasks scheduled on a four-core processor, and each task is assigned two processors, as shown in Fig.~1. Job $J_i$ misses its deadline at time $t=11$ under any protocol with local execution of requests, as shown in Fig. 1(b). In contrast, $J_i$ is schedulable if the request of resource $\ell_1$ is conducted on processor $\wp_2$, as depicted in Fig. 1(c).

\begin{figure}[t]
\centering
\setlength{\belowcaptionskip}{-15pt}

\subfigure[The structures of two DAG tasks with resources $\ell_1$ (\textcolor{red}{red}) and $\ell_2$ (\textcolor{blue}{blue}).]{
\includegraphics[width=0.48\textwidth]{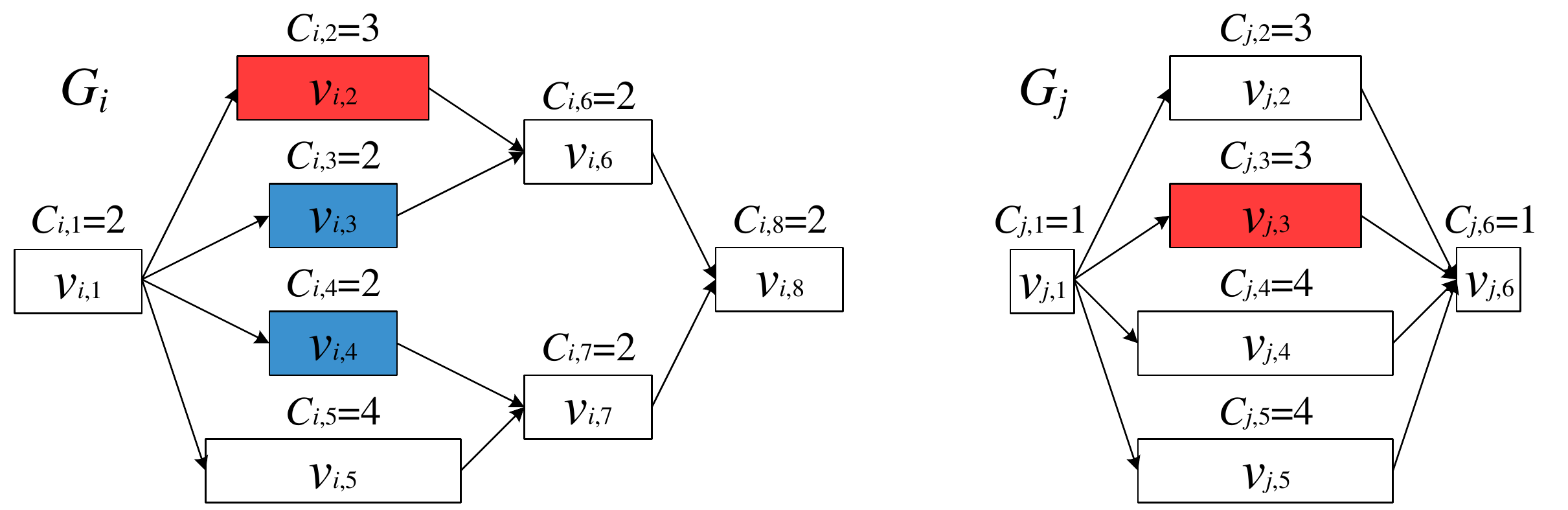}}
%\subfigure[Example schedule with local execution of resources.]{
%\includegraphics[width=0.48\textwidth]{example_2}}
\subfigure[Example schedule of DPCP-p with $\ell_1$ being assigned to $\wp_2$.]{
\includegraphics[width=0.48\textwidth]{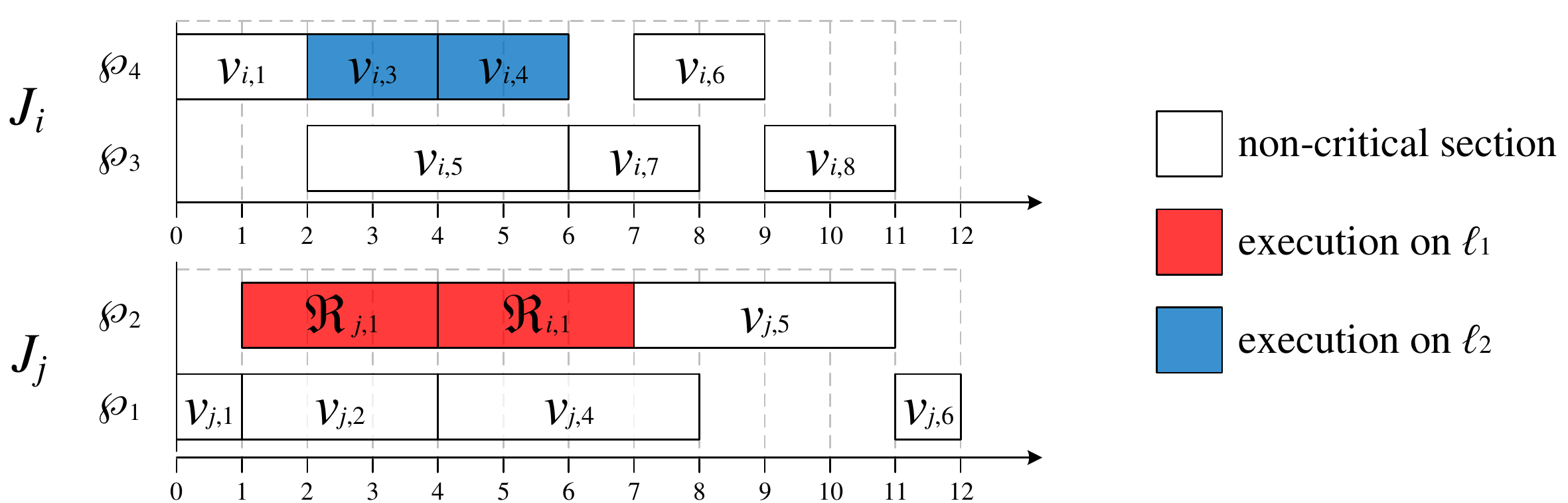}}
\caption{Example schedule of two DAG tasks.}
%\lable{fig:example}
\end{figure}

\begin{lemma}
\label{lem:pi_blocking}
Under DPCP-p, a request can be blocked by lower-priority requests at most once.
\begin{proof}
We prove by contradiction. Since each local resource is used only by a task, we consider global-resource-requests. Suppose that a request $\Re_{i,q}$ ($\ell_q\in\Phi^G$) on a processor $\wp_k$ is blocked by at least two lower-priority requests $\Re_{a,u}$ and $\Re_{b,v}$ ($\pi_a<\pi_i$, $\pi_b<\pi_i$). Let $t_{i,s}$ and $t_{i,f}$ be the time when $\Re_{i,q}$ arrives and finishes respectively. Let $t_{a,r}$ and $t_{b,r}$ be the time when $\Re_{a,u}$ and $\Re_{b,v}$ are granted the locks respectively.  Without loss of generality, we assume that $t_{a,r}<t_{b,r}$.

While $\Re_{i,q}$ is pending at some time $t\in [t_{i,r}, t_{i,f}]$, the processor ceiling $\Pi^{\wp}_k(t)\geq\pi^H+\pi_i$ according to the priority ceiling mechanism. Since $\Re_{i,q}$ can be blocked by $\Re_{a,u}$, the priority ceiling of $\ell_u$ is larger than $\pi^H+\pi_i$, i.e., $\Pi_u\geq\pi^H+\pi_i$. Thus, $\Pi^{\wp}_k(t)\geq\pi^H+\pi_i$ during $t\in [\min(t_{i,s},t_{a,r}),t_{i,f}]$. Further, by hypothesis, $\Re_{i,q}$ is blocked by $\Re_{b,v}$, then $\Re_{b,v}$ must be granted the lock at some time $t\in (t_{a,r},t_{i,f})$. Accoring to the priority ceiling mechanism, the effective priority of $\Re_{b,v}$ must be larger than the processor ceiling at time $t$, i.e., $\pi^E_i(t)=\pi^H+\pi_b>\Pi^{\wp}_k(t)\geq\pi^H+\pi_i$. Thus, $\pi_b>\pi_i$. Contradiction.
\end{proof}
\end{lemma}

\section{Worst-Case Response Time Analysis}
\label{sec:rta}

We derive an upper bound of the WCRT of an arbitrary path of $\tau_i$ using the fixed-point Response-Time Analysis (RTA) in this section. Let $r_i$ be the WCRT of an arbitrary path $\lambda_i$, then $R_i$ can be upper bounded by the maximum of the WCRTs of the paths, that is

\begin{equation}
\label{eq:wcrt}
R_i=\max\{r_i\}.
\end{equation}

To upper bound $r_i$, we classify the delays of a path into four categories as follows. %in Sect.~\ref{sec:blocking_interference}. 

%A task $\tau_i$ is considered schedulable if $R_i\leq D_i$, and the taskset is considered schedulable if all the tasks are schedulable.

\subsection{Blocking and Interference}
\label{sec:blocking_interference}
First, we consider a global-resource-request $\Re_{j,q}$ of a job $J_j$ ($i\neq j$, $\ell_q\in\Phi^G$) that causes $\lambda_i$ to incur
\begin{itemize}
\item \emph{inter-task blocking}, if an agent on behalf of $\Re_{j,q}$ is executing on some processor $\wp_k$ while $\lambda_i$ is suspended on a global resource $\ell_u\in\Phi^G$ on $\wp_k$.
\end{itemize}

Second, a vertex $v_{i,y}$ of $\tau_i$ that is not on $\lambda_i$ (i.e., $v_{i,y}\notin\lambda_i$) causes $\lambda_i$ to incur
\begin{itemize}
\item \emph{intra-task blocking}, if $v_{i,y}$ is holding a local resource $\ell_q\in\Phi^L$ and scheduled while $\lambda_i$ is suspended on $\ell_q$, or if an agent is executing on behalf of $v_{i,y}$ on some processor $\wp_k$ while $\lambda_i$ is suspended on a global resource on $\wp_k$; and
\item \emph{intra-task interference}, if $v_{i,y}$ is executing a non-critical section or a local-resource-request while $\lambda_i$ is ready and not executing.
\end{itemize}

Third, a global-resource-request from another job or from a vertex that is not on $\lambda_i$ causes $\lambda_i$ to incur
\begin{itemize}
\item \emph{agent interference}, if an agent on behalf of the request is executing while $\lambda_i$ is (i) ready and not executing, or (ii) suspended on a local resource and the resource holder is not scheduled (i.e., preempted by the agent of the request).
\end{itemize}

Notabaly, the defined delays are mutually exclusive, i.e., at any point in time, a vertex or an agent can cause a path to incur at most one type of delay. This is essential to minimize over-counting in the blocking time analysis.
For example in Fig. 1(b), at any time during $t=[2,4]$, $\Re_{j,1}$ only causes path $(v_{i,1},v_{i,2},v_{i,6},v_{i,8})$ to incur inter-task blocking, $v_{i,3}$ only causes path $(v_{i,1},v_{i,4},v_{i,7},v_{i,8})$ to incur intra-task blocking, $v_{j,2}$ only causes path $(v_{j,1},v_{j,4},v_{j,6})$ to incur intra-task interference, and $\Re_{j,1}$ only causes path $(v_{j,1},v_{j,5},v_{j,6})$ to incur agent interference. It is also noted that a path can incur multiple types of delays at a time. For instance, at any time during $t=[1,4]$, path $(v_{j,1},v_{j,5},v_{j,6})$ incurs intra-task interference and agent interference due to $v_{j,2}$ and $\Re_{j,1}$ respectively. 

Based on these definitions, we derive an upper bound on $r_i$ in Theorem~\ref{theorem:rta-path}. In preparation, we use $B_i$ to denote the workload of the other tasks that causes $\lambda_i$ to incur inter-task blocking. Analogously, let $b_i$ and $I_i^{\text{intra}}$ denote the workloads of the vertices of $\tau_i$ not on $\lambda_i$ that cause $\lambda_i$ to incur intra-task blocking and intra-task interference, respectively. Let $I_i^A$ denote the workload of the agents that causes $\lambda_i$ to incur agent interference. %The values of the defined variables are nonnegative. 
These open variables will be bounded in Sect. \ref{sec:blocking} and \ref{sec:interference}.

\begin{theorem}
\label{theorem:rta-path}
%The WCRT of a path $\lambda_i$ can be upper bounded by
%\begin{equation}
$r_i\leq\mathcal{L}(\lambda_i)+B_i+b_i+(I_i^{\text{intra}}+I_i^A)/m_i$.
%\end{equation}
\end{theorem}
\begin{proof}
While $\lambda_i$ is pending, it is either (I) executing, (II) suspended and executing on global resources, (III) ready and not executing, (IV) suspended and not executing on any global resource. By definition, the duration of (I) and (II) can be bounded by $\mathcal{L}(\lambda_i)$.

\emph{For case (III)}. The workload executed on $\wp(\tau_i)$ can be from (i) the vertices of $\tau_i$ not on $\lambda_i$ (i.e., intra-task interference), and (ii) the agents on behalf of the vertices not on $\lambda_i$ (i.e., agent interference). By definition, the workload of (i) can be upper-bounded by $I_i^{\text{inter}}$, and the workload of (ii), denoted by $\hat{I}_i^A$, is a part of $I_i^A$. 

\emph{For case (IV)}. If $\lambda_i$ is suspended on a local resource $\ell_q$, then $\lambda_i$ is either (iii) waiting a vertex of $\tau_i$ not on $\lambda_i$ to release $\ell_q$ (i.e., intra-task blocking), or (iv) waiting the agents that preempted the resource holder to finish (i.e., agent interference). If $\lambda_i$ is suspended on a global resource on a processor $\wp_k$, then it can be delayed by (v) an agent on behalf of another task on $\wp_k$ (i.e., inter-task blocking), or (vi) an agent on behalf of a vertex of $\tau_i$ not on $\lambda_i$ on $\wp_k$ (i.e., intra-task blocking). By definition, the duration of (iii) and (vi) is $b_i$, and the duration of (v) is $B_i$. Further, for case (iv), we let the workload of the agents be $\check{I}_i^A$.

\emph{Total durations of (I) - (IV)}. In (III) and (IV)-(iv), there is at least a vertex ready and not executing, thus none of the $m_i$ processors is idle according to work-conserving scheduling. Let the duration of (III) and (IV)-(iv) be $Y_i$, then $I_i^{\text{intra}}+\hat{I}_i^A+\check{I}_i^A=Y_i\cdot m_i$. By definition, $\hat{I}_i^A+\check{I}_i^A\leq I_i^A$. Hence, $Y_i\leq(I_i^{\text{intra}}+I_i^A)/m_i$. Summing up (I) - (IV), we have $r_i\leq\mathcal{L}(\lambda_i)+B_i+b_i+(I_i^{\text{intra}}+I_i^A)/m_i$.
\end{proof}

\subsection{Upper Bounds on Blockings}
\label{sec:blocking}
We begin with the analysis of inter-task blocking. To derive an upper bound on $B_i$, we first derive an upper bound on the response time of a global-resource-request. 

In preparation, let $\eta_j(L)$ denote the maximum number of jobs of a task $\tau_j$ during a time interval of length $L$. It has been well studied that $\eta_j(L)\leq\lceil (L+R_j)/T_j \rceil$. 
Further, let $\gamma_{i,q}(L)$ be the cumulative length of the requests from higher-priority tasks of $\tau_i$ to the global resources that are assigned on the same processor as $\ell_q\in\Phi^G$ during a time interval of length $L$. Since there are $\eta_h(L)$ jobs of each higher-priority task $\tau_h$ ($\pi_h>\pi_i$) during a time interval length of $L$, and each job $J_h$ uses resource $\ell_q$ for a time of at most $N_{h,q}\cdot L_{h,q}$, summing up the workload of all the higher-priority requests we have 
\begin{equation}
\gamma_{i,q}(L)\leq\underset{\pi_h>\pi_i\wedge\ell_u\in\Phi^{\wp}(\ell_q)}{\sum} \eta_h(L) \cdot N_{h,u} \cdot L_{h,u}.
\end{equation}

Let $W_{i,q}$ be the response time of a request from $\lambda_i$ to a global resource $\ell_q\in\Phi^G$. We bound $W_{i,q}$ according to the following lemma. 
\begin{lemma}
\label{lem:res-hold-time}
$W_{i,q}$ can be upper bounded by the least positive solution, if one exists, of the following equation.
\begin{equation}
W_{i,q}=L_{i,q}+\underset{\ell_u\in\Phi^{\wp}(\ell_q)}{\sum}(N_{i,u}-N_{i,u}^{\lambda})\cdot L_{i,u} +\beta_{i,q}+\gamma_{i,q}(W_{i,q}).
\end{equation}
Where, $\beta_{i,q}=\max\{L_{j,u}|\pi_j<\pi_i\wedge\ell_u\in\Phi^{\wp}(\ell_q)\wedge\Pi_u\geq\pi^H+\pi_i\}.$
\end{lemma}
\begin{proof}
Under DPCP-p, a global-resource-request $\Re_{i,q}$ has an effective priority higher than $\pi^H$. Thus, while $\Re_{i,q}$ is pending, only the global-resource-requests can execute. Since global-resource-requests are scheduled by their priorities, $\Re_{i,q}$ may wait for (i) at most one lower-priority request to a global resource with priority ceiling no less than $\pi^H+\pi_i$ on the processor, (ii) intra-task requests from the vertices not on $\lambda_i$ to the global resources on the processor, and (iii) higher-priority requests to the global resources on the processor. 

Be definition, (i) can be bounded by $\beta_{i,q}$, and (ii) can be bounded by $\sum_{\ell_u\in\Phi^{\wp}(\ell_q)}(N_{i,u}-N_{i,u}^{\lambda})\cdot L_{i,u}$. By the definition of $\gamma_{i,q}(L)$, (iii) can be bounded by $\gamma_{i,q}(W_{i,q})$. In addition, $\Re_{i,q}$ executes at most $L_{i,q}$. We claim the lemma by summing up the respective bounds.
\end{proof}

With Lemma~\ref{lem:res-hold-time} in place, we are ready to upper bound $B_i$.
\begin{lemma}
\label{lem:g-blocking}
$B_i\leq\sum_{\wp_k\in\wp}\min(\varepsilon_i^k, \zeta_i^k)$,
where,
\begin{equation}
\label{eq:varepsilon}
\varepsilon_i^k=\underset{\ell_q\in\Phi^G\cap\Phi(\wp_k)}{\sum} \left( \beta_{i,q}+\gamma_{i,q}(W_{i,q}) \right) \cdot N_{i,q}^{\lambda},
\end{equation}
and
\begin{align}
\label{eq:zata}
\zeta_i^k=\sum_{\tau_j\neq\tau_i} \sum_{\ell_q\in\Phi^G\cap\Phi(\wp_k)} \eta_j(r_i)\cdot N_{j,q} \cdot L_{j,q}.
\end{align}
\end{lemma}
\begin{proof}
Each time $\lambda_i$ requests a global resource $\ell_q\in\Phi^G$ on $\wp_k$, it can be blocked by (i) at most one lower-priority request and (ii) all higher-priority requests. Analogous to the proof in Lemma~\ref{lem:res-hold-time}, (i) can be bounded by $\beta_{i,q}$, and (ii) can be bounded by $\gamma_{i,q}(W_{i,q})$. 
%Thus, $\beta_{i,q}+\gamma_{i,q}(W_{i,q})$ bounds the workloads of the other tasks that cause $\lambda_i$ to incur inter-task blocking for a single request.
Since $\lambda_i$ requests each global resource $\ell_q$ at most $N_{i,q}^{\lambda}$ times, the workload of the other tasks that cause $\lambda_i$ to incur inter-task blocking on $\wp_k$ can be bounded by $\varepsilon_i^k$ in Eq.~(\ref{eq:varepsilon}).

Further, each other task $\tau_j$ ($j\neq i$) has at most $\eta_j(r_i)$ jobs before $\lambda_i$ finishes, and each job uses a resource $\ell_q$ for a time of at most $N_{j,q}\cdot L_{j,q}$. Thus, the workload of the other tasks for the global resources on $\wp_k$ can be bounded by $\zeta_i^k$ in Eq.~(\ref{eq:zata}).
We claim the lemma by summing up the minimum of $\varepsilon_i^k$ and $\zeta_i^k$ for all processors. 
\end{proof}

Next, we derive an upper bound for intra-task blocking. For brevity, let $\sigma_{i,k}=\min(1, \sum_{\ell_u\in\Phi(\wp_k)}N^{\lambda}_{i,u})$. Intuitively, $\sigma_{i,k}=1$ if there is a vertex on path $\lambda_i$ requests a global resource $\ell_q$ on processor $\wp_k$, and $\sigma_{i,k}=0$ otherwise.
\begin{lemma}
\label{lem:intra-task-blocking}
$b_i\leq\underset{\ell_q\in\Phi^L\cap\Phi(\tau_i)}{\sum}b_{i,q}^L+\sum_{\wp_k\in\wp}b_i^G$, where,
\begin{equation}
\label{eq:intra-b-local}
b_{i,q}^L=\min(1,N_{i,q}^{\lambda})\cdot(N_{i,q}-N_{i,q}^{\lambda})\cdot L_{i,q},
\end{equation}
and, 
\begin{equation}
\label{eq:intra-b-global}
b_i^G=\sigma_{i,k}\cdot\sum_{\ell_q\in\Phi(\wp_k)}(N_{i,q}-N_{i,q}^{\lambda})\cdot L_{i,q}. 
\end{equation}
\end{lemma}
\begin{proof}
By definition, $\lambda_i$ incurs intra-task blocking on a local resource $\ell_q\in\Phi^L$ only if it requests $\ell_q$. Clearly, $\min(1, N_{i,q}^{\lambda})=1$ if $\lambda_i$ requests $\ell_q$, and $\min(1, N_{i,q}^{\lambda})=0$ otherwise. Given that the vertices of $\tau_i$ not on $\lambda_i$ can execute on a resource $\ell_q$ for a total of at most $(N_{i,q}-N_{i,q}^{\lambda})\cdot L_{i,q}$, $\lambda_{i,q}$ incurs intra-task blocking for at most $b_{i,q}^L$, as in Eq.~(\ref{eq:intra-b-local}).

Moreover, $\lambda_i$ incurs intra-task blocking on a global resource on some processor $\wp_k$ only if it requests some global resource on $\wp_k$. By definition, $\sigma_{i,k}=1$ if $\lambda_i$ requests some global resource on $\wp_k$, and $\sigma_{i,k}=0$ otherwise. Thus, the workload that cause $\lambda_i$ to incur intra-task blocking on $\wp_k$ can be bounded by summing up $(N_{i,q}-N_{i,q}^{\lambda})\cdot L_{i,q}$ for all the global resources on $\wp_k$, i.e., $b_i^{\wp}$, as in Eq.~(\ref{eq:intra-b-global}).

Thus, $b_i$ can be bounded by summing up (i) $b_{i,q}^L$ for all local resource in $\Phi(\tau_i)$, and (ii) $b_i^{\wp}$ for all processors.
\end{proof}

\subsection{Upper Bounds on Interference}
\label{sec:interference}
Next, we derive upper bounds for the intra-task interference and the agent interference. First, the intra-task interference of $\lambda_i$ can be upper bounded by the workload of the non-critical sections and the local-resource-requests of the vertices of $\tau_i$ that are not on $\lambda_i$.
\begin{lemma}
\label{lem:intra-task-interf}
$I_i^{\text{intra}}\leq\underset{v_{i,x}\notin\lambda_i}{\sum}C_{i,x}^{\prime}+\underset{\ell_q\in\Phi^L}{\sum}(N_{i,q}-N_{i,q}^{\lambda})\cdot L_{i,q}$.
\end{lemma}
\begin{proof}
By definition, $I_i^{\text{intra}}$ consists of the workload of (i) the non-critical sections and (ii) the local-resource-requests from the vertices of $\tau_i$ that are not on $\lambda_i$. From the task model, (i) and (ii) are bounded by $\sum_{v_{i,x}\notin\lambda_i}C_{i,x}^{\prime}$ and $\sum_{\ell_q\in\Phi^L}(N_{i,q}-N_{i,q}^{\lambda})\cdot L_{i,q}$, respectively. Thus, $I_i^{\text{intra}}$ can be bounded by the total of (i) and (ii).
\end{proof}

For each global resource on $\Phi^{\wp}(\tau_i)$, the agent interference of $\lambda_i$ consists of the agent workload of the vertices that are not on $\lambda_i$.
\begin{lemma}
\label{lem:agent-interf}
$I_i^A\leq \underset{\ell_q\in\Phi^G\cap\Phi^{\wp}(\tau_i)}{\sum} (I^A_{i,q}+\breve{I}^A_{i,q})$,
where, 
\begin{equation}
I^A_{i,q}=\sum_{\tau_j\neq\tau_i}\eta_j(r_i)\cdot N_{j,q}\cdot L_{j,q},
\end{equation}
and,
\begin{equation}
\breve{I}^A_{i,q}=(N_{i,q}-N^{\lambda}_{i,q})\cdot L_{i,q}.
\end{equation}
\end{lemma}
\begin{proof}
While $\lambda_i$ is pending, the other tasks can request a resource $\ell_q$ for at most $I_{i,q}^A$, and the vertices of $\tau_i$ not on $\lambda_i$ can execute on $\ell_q$ for at most $\breve{I}_{i,q}^A$. Thus, the agent interference of $\lambda_i$ can be bounded by summing up $I_{i,q}^A+\breve{I}_{i,q}^A$ for all the global resources on $\Phi^{\wp}(\tau_i)$.
\end{proof}

Now that we bounded all the variables in Theorem~\ref{theorem:rta-path}, thus the WCRT of task $\tau_i$ can be bounded according to Eq.~(\ref{eq:wcrt}) by calculating the WCRTs of all paths of $\tau_i$. 

\section{Task and Resource Partitioning}
\label{sec:partitioning}

According to the schedulability analysis in Sect.~\ref{sec:rta}, the WCRT of a task can be determined only if the tasks and the global resources are partitioned. In this section, we present a partitioning algorithm to iteratively assign tasks and resources.

 For ease of description, we consider the processors assigned to each task as a cluster. Accordingly, we use $\wp^C_x$ to denote the $x$th cluster ($x\leq m$). The capacity of $\wp^C_x$, denoted by $\tilde{U}^{\text{cluster}}_x$, is equal to the number of the processors in $\wp^C_x$. The utilization of $\wp^C_x$, denoted by $U^{\text{cluster}}_x$, is the total of the utilizations of the task and the resources assigned to $\wp^C_x$, where the utilization of a resource $\ell_q$ is defined as $u^{\Phi}_q=\sum_{\tau_j\in\tau}\frac{N_{j,q}\cdot L_{j,q}}{T_j}$. The total utilization of the global resources assigned to a processor $\wp_k$ is denoted by $u^{\wp}_k$, i.e., $u^{\wp}_k=\sum_{\ell_q\in\Phi(\wp_k)}u^{\Phi}_q$. The utilization slack of a cluster $\wp^C_x$ is defined by $\tilde{U}^{\text{cluster}}_x-U^{\text{cluster}}_x$. A cluster is infeasible if $U^{\text{cluster}}_x>\tilde{U}^{\text{cluster}}_x$. 

Each task $\tau_i$ is initially assigned $\lceil \frac{C_i-\mathcal{L}_i^{\ast}}{D_i-\mathcal{L}_i^{\ast}}\rceil$ processors, and the global resources are partitioned according to the Worst-Fit Decreasing (WFD) heuristic, as shown in Algorithm 1. 
Intuitively, the global resource with the highest utilization is assigned to the processor with the lowest resource utilization in the cluster with maximum utilization slack, as shown in Algorithm 2.
The schedulability analysis is performed from the task with highest base priority. If there is a task unschedulable, then we assign an additional processor, if one exists, to that task. Since the capacity of the cluster is updated when an additional processor is assigned, we re-assign global resources and perform schedulability tests accordingly. The partitioning process runs at most $m-2n$ rounds for systems containing only heavy tasks.

\begin{algorithm}[t]
    \caption{Task and Resource Partitioning}
%\begin{spacing}{1.1}
\begin{algorithmic}[1]
\small
\Require
the tasks $\tau$, the processors $\wp$, and the resources $\Phi$
\Ensure
the schedulability of the system

\State {\textbf{for} $\forall\tau_i\in\tau$ \textbf{do}}
\State {\ \ \ \ \textbf{if} there are $\lceil (C_i-\mathcal{L}_i^{\ast})/(D_i-\mathcal{L}_i^{\ast})\rceil$ processors unassigned \textbf{then}}
\State {\ \ \ \ \ \ \ \ assign $\lceil (C_i-\mathcal{L}_i^{\ast})/(D_i-\mathcal{L}_i^{\ast})\rceil$ processors to $\tau_i$}
\State {\ \ \ \ \textbf{else}}
\State {\ \ \ \ \ \ \ \ \textbf{return} {unschedulable}}
%\State {\ \ \ \ \textbf{end if}}
%\State {\textbf{end for}}

\State {\textbf{while} true \textbf{do}}

\State {\ \ \ \ \textbf{if} WFD\_Resource($\Phi^G$, $\wp$) is infeasible \textbf{then}}
\State {\ \ \ \ \ \ \ \ \textbf{return} {unschedulable}}
%\State {\ \ \ \ \textbf{end if}}

\State {\ \ \ \ \textbf{for} $\forall\tau_i\in\tau$ in decreasing order of priority \textbf{do}}
\State {\ \ \ \ \ \ \ \ \textbf{if} $\text{WCRT}(\tau_i)>D_i$ \textbf{then}}
\State {\ \ \ \ \ \ \ \ \ \ \ \ \textbf{if} there is a processor unassigned \textbf{then}}
\State {\ \ \ \ \ \ \ \ \ \ \ \ \ \ \ \ assign one more processor to $\tau_i$}
\State {\ \ \ \ \ \ \ \ \ \ \ \ \ \ \ \ rollback of the global resource assignment}
\State {\ \ \ \ \ \ \ \ \ \ \ \ \ \ \ \ \textbf{break}\ \ \ \ \ // i.e., go to line 9}
\State {\ \ \ \ \ \ \ \ \ \ \ \ \textbf{else}}
\State {\ \ \ \ \ \ \ \ \ \ \ \ \ \ \ \ \textbf{return} {unschedulable}}
%\State {\ \ \ \ \ \ \ \ \ \ \ \ \textbf{end if}}
%\State {\ \ \ \ \ \ \ \ \textbf{end if}}
%\State {\ \ \ \ \textbf{end for}}
\State {\ \ \ \ \textbf{return} {schedulable}}
%\State {\textbf{end while}}
    \label{alg:1}
\end{algorithmic}
%\end{spacing}
\end{algorithm}

\begin{algorithm}[t]
    \caption{WFD\_Resources}
%\begin{spacing}{1.1}
\begin{algorithmic}[1]
\small
\Require
the global resources $\Phi^G$, and the processors $\wp$
\Ensure
the feasibility of the global resource allocation

\State{sort $\Phi^G$ in non-increasing order of utilization}
\State {\textbf{for} $\forall\tau_i\in\tau$ \textbf{do}}
\State {\ \ \ \ $U^{\text{cluster}}_i=m_i$}
%\State {\textbf{endfor}}

\State {\textbf{for} $\forall\ell_q\in\Phi^G$ \textbf{do}}
\State {\ \ \ \ select the cluster $\wp^C_x$ with the maximum value of $\tilde{U}^{\text{cluster}}_x-U^{\text{cluster}}_x$}
\State {\ \ \ \ \textbf{if} $U^{\text{cluster}}_x+u^{\Phi}_q>\tilde{U}^{\text{cluster}}_x$ \textbf{then}}
\State {\ \ \ \ \ \ \ \ \textbf{return} infeasible allocation}
\State {\ \ \ \ \textbf{else}}
\State {\ \ \ \ \ \ \ \ assign $\ell_q$ to processor $\wp_k$, s.t., $u^{\wp}_k=\min\{u^{\wp}_a|\wp_a\in\wp^C_x\}$}
\State {\ \ \ \ \ \ \ \ $U^{\text{cluster}}_x=U^{\text{cluster}}_x+u^{\Phi}_q$}
\State {\Return {feasible allocation}}
    \label{alg:3}
\end{algorithmic}
%\end{spacing}
\end{algorithm}

\section{Discussions}
\label{sec:discussions}

Although we focus on heavy tasks in this paper, the DPCP-p approach can be extended to support light tasks. First, light tasks are treated as sequential tasks under federated scheduling, thus the original DPCP can be used to handle resource sharing between them. Further, since each heavy task is exclusively assigned a cluster of processors, the delays between heavy and light tasks are only due to global resources. According to the definitions in Sect.~\ref{sec:blocking_interference}, such delays can be captured by inter-task blocking and agent interference. According to Lemma~\ref{lem:g-blocking} and Lemma~\ref{lem:agent-interf}, bounding both inter-task blocking and agent interference does not distinguish between heavy and light tasks. Thus, the delays between heavy and light tasks can be analyzed using the analysis framework as presented in Sect.~\ref{sec:rta}. Notably, handling light tasks with shared resources optimally under federated scheduling remains as an open problem. 

Further, we assume that the maximum number of requests of each vertex $N_{i,x,q}$ is known. This is possible in some real-life applications when the maximum number of requests of each vertex can be pre-determined. Thus we can derive a more accurate blocking bound by using the exact number of requests on a path $\lambda_i$, i.e., $N_{i,q}^{\lambda}=\sum_{v_{i,x}\in\lambda_i}N_{i,x,q}$, rather than enumerating the value of $N_{i,q}^{\lambda}$ from $[0,N_{i,q}]$~\cite{DBLP:journals/tpds/DinhLAGL18}. The tradeoff is more calculations to enumerate all paths of the task in analysis. Notably, the presented analysis applies to the prior model \cite{DBLP:journals/tpds/DinhLAGL18,DBLP:conf/dac/JiangGLY19} by using the key-path-oriented analysis~\cite{DBLP:conf/dac/JiangGLY19}.
   
The blocking-time analysis can be further improved by modern analysis techniques, e.g., the Linear-Programming-based (LP-based) analysis in~\cite{DBLP:conf/rtas/Brandenburg13}. However, we have no evidence on how the LP-based analysis \cite{DBLP:conf/rtas/Brandenburg13} can be applied for this scenario yet. Thus, we first establish the fundamental analysis framework in this paper and remain fine-grained analysis as future work. 

%Instead, we first establish the fundamental analysis framework and remain more fine-grained anaysis as future work.

\section{Empirical Comparisons}
\label{sec:experiments}

In this section, we conduct schedulability experiments to evaluate the DPCP-p approach using synthesized heavy tasks.

\subsection{Experimental Setup}
Multiprocessor platforms with $m\in\{8, 16, 32\}$ unispeed processors and $n_r$, ranging over $[2,4]$, $[4,8]$ or $[8,16]$, shared resources were considered. For each $m$, we generated the total utilizations of the tasksets from 1 to $m$ in steps of 0.05$m$. The task utilizations of a taskset were generated according to the RandFixedSum algorithm~\cite{EmbersonWATERS2010} ranging over $(1, 2U^{\text{avg}}]$, where $U^{\text{avg}}\in\{1.5,2\}$ is the average utilization of the tasks. The base priority of the tasks was assigned by the Rate Monotonic (RM) heuristic. The number of tasks $n$ was determined by the chosen $U^{\text{avg}}$ and the total utilization of the taskset.

For each task $\tau_i$, the DAG structure was generated by the Gr{\'{e}}gory Erd{\"{o}}s-R{\'{e}}nyi algorithm~\cite{DBLP:conf/simutools/CordeiroMPTVW10}, where the number of vertices $|V_i|$ was randamly chosen in $[10, 100]$, and the probability of an edge between any two vertices was set to 0.1. Task period $T_i$ was randomly chosen from log-uniform distributions ranging over $[10ms, 1000ms]$, and $C_i$ was computed by $U_i \cdot T_i$. $\tau_i$ uses each resource in a probability $p^r=\{0.5,0.75,1\}$. If $\tau_i$ used $\ell_q$, the maximum number of requests $N_{i,q}$ was randomly chose from $[1,25]$ or $[1,50]$, and the maximum critical section length $L_{i,q}$ was chosen in $[15\mu s, 50\mu s]$ or $[50\mu s, 100\mu s]$. The WCET of each vertex $C_{i,x}$ and the maximum number of requests in each vertex $N_{i,x,q}$ were randamly determined such that $C_i=\sum_{v_{i,x}\in V_i}C_{i,x}$ and $N_{i,q}=\sum_{v_{i,x}\in V_i}N_{i,x,q}$. To ensure plausibility, we enforced that $\mathcal{L}^{\ast}_i<D_i/2$ and $C_{i,x}\geq\sum_{\ell_q\in\Phi}N_{i,x,q}\cdot L_{i,q}$. The combination of the parameters consists of 216 experimental scenarios. %Next, we report the main trends of the experimental results.

\subsection{Baselines}
We compare DPCP-p with existing locking protocols, denoted by SPIN-SON~\cite{DBLP:journals/tpds/DinhLAGL18} and LPP~\cite{DBLP:conf/dac/JiangGLY19}, under federated scheduling (there is no study on locking protocols for the other state-of-the-art scheduling approaches in the literature, for which we will discuss in Sect.~\ref{sec:related-work}). For DPCP-p, we use DPCP-p-EP to denote the analysis as presented in Sect.~\ref{sec:rta} by enumerating all paths, and use DPCP-p-EN to denote the analysis by enumerating $N_{i,q}^{\lambda}$ from 0 to $N_{i,q}$ for $\forall\ell_q\in\Phi$ as in~\cite{DBLP:journals/tpds/DinhLAGL18,DBLP:conf/dac/JiangGLY19}. We also use FED-FP to denote a hypothesis baseline without considering shared resources under federated scheduling~\cite{DBLP:conf/ecrts/LiCALGS14}.
\iffalse
It is worth noting that there is significant progress on the scheduling of parallel real-time tasks that are not based on federated scheduling, e.g., partitioned~\cite{DBLP:conf/rtss/CasiniBNB18a}, semi-partitioned~\cite{DBLP:conf/ipps/BonifaciDM17}, global~\cite{DBLP:conf/rtns/FonsecaNN17}, and decomposition-based scheduling~\cite{DBLP:conf/rtss/JiangLGW16}. However,  so far as we know. Moreover, since locking requests are performed by each vertex, the locking protocols that are originally used for sequential tasks, e.g., \cite{DBLP:conf/rtss/RajkumarSL88,DBLP:conf/icdcs/Rajkumar90}, might be used. However, no work on the corresponding analysis has been established in the literature. Therefore, we focus on the performance of existing locking protocols for parallel real-time tasks under federated scheduling.
\fi
\subsection{Results}
Fig. 2 shows acceptance ratios of the tested approaches with increasing normalized utilization, where Fig. 2(b) and (d) include more resource contentions compared to Fig. 2(a) and (c). It is shown that DPCP-p-EP consistently schedules more tasksets than SPIN-SON and LPP. In particular, the advantage of the DPCP-p approach is more significant for heavy resource-contentions as shown in Fig. 2(b) and (d), while SPIN-SON appears to be competitive for light resource-contentions as shown in Fig. 2(a) and (c). 

For brevity, we use the notations of \emph{dominance} and \emph{outperformance}\footnote{For an experimental scenario, algorithm $A$ is said to outperform algorithm $B$ if algorithm $A$ scheduled more task sets than algorithm $B$, or dominate algorithm $B$ if its acceptance ratio is higher than that of algorithm $B$ at some tested points and never lower than that of algorithm $B$ at any tested point.} to summarize the main trends of the results in Table 2 and 3. It is shown that the DPCP-p approach improves upon SPIN-SON and LPP significantly. In particular, DPCP-p-EP outperforms in all scenarios, and it dominates in more than 99\% scenarios. Similarly, DPCP-p-EN dominates and outperforms more often than less.
%However, the results only compare the performance of locking protocols under federated scheduling and should not be understood as an absolute ranking for parallel task systems. 
\begin{figure}
\setlength{\belowcaptionskip}{-15pt}
\begin{minipage}{0.48\linewidth}
\centering
\includegraphics[width=4cm]{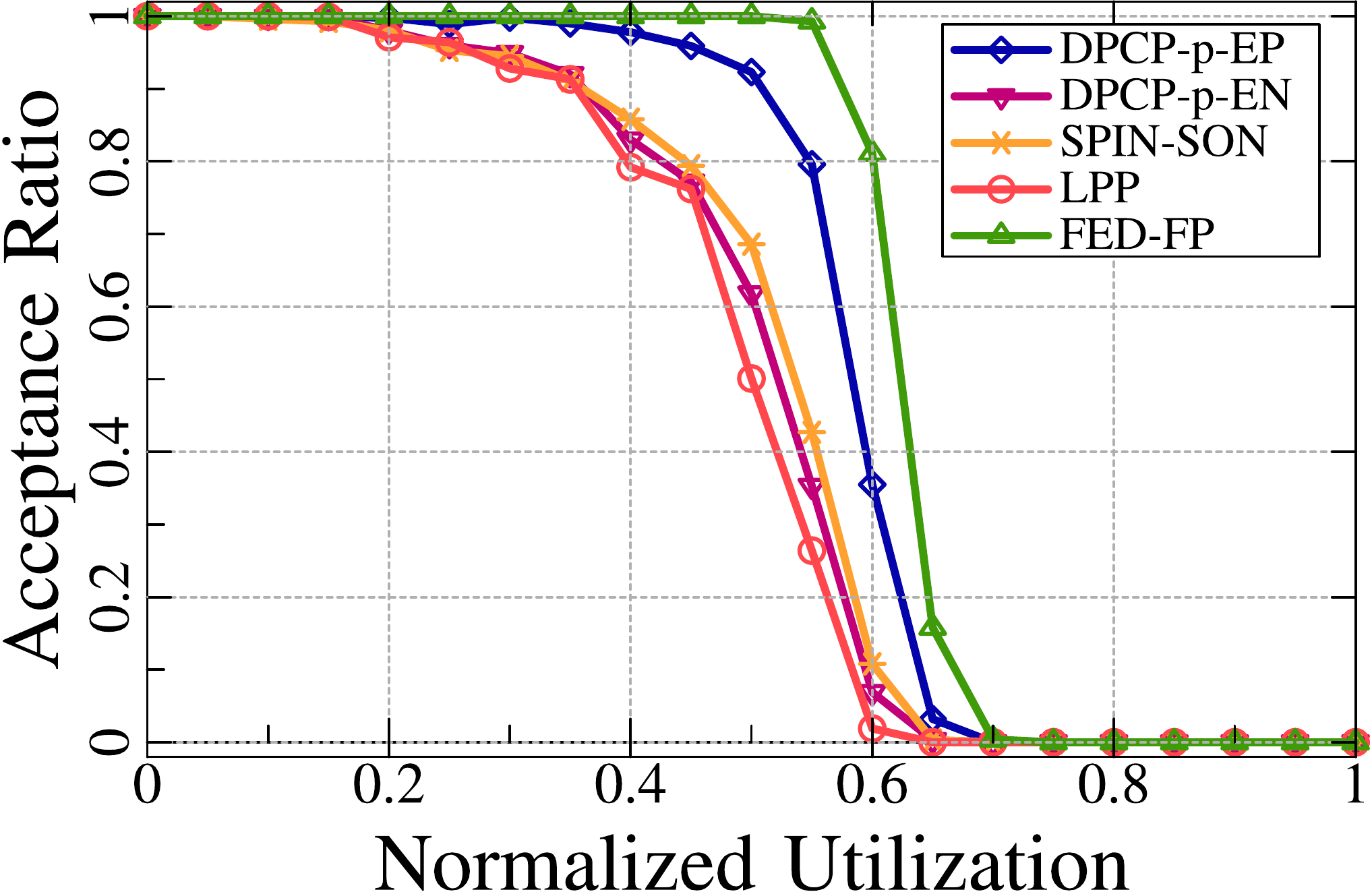}
\end{minipage}
\begin{minipage}{0.48\linewidth}
\centering
\includegraphics[width=4cm]{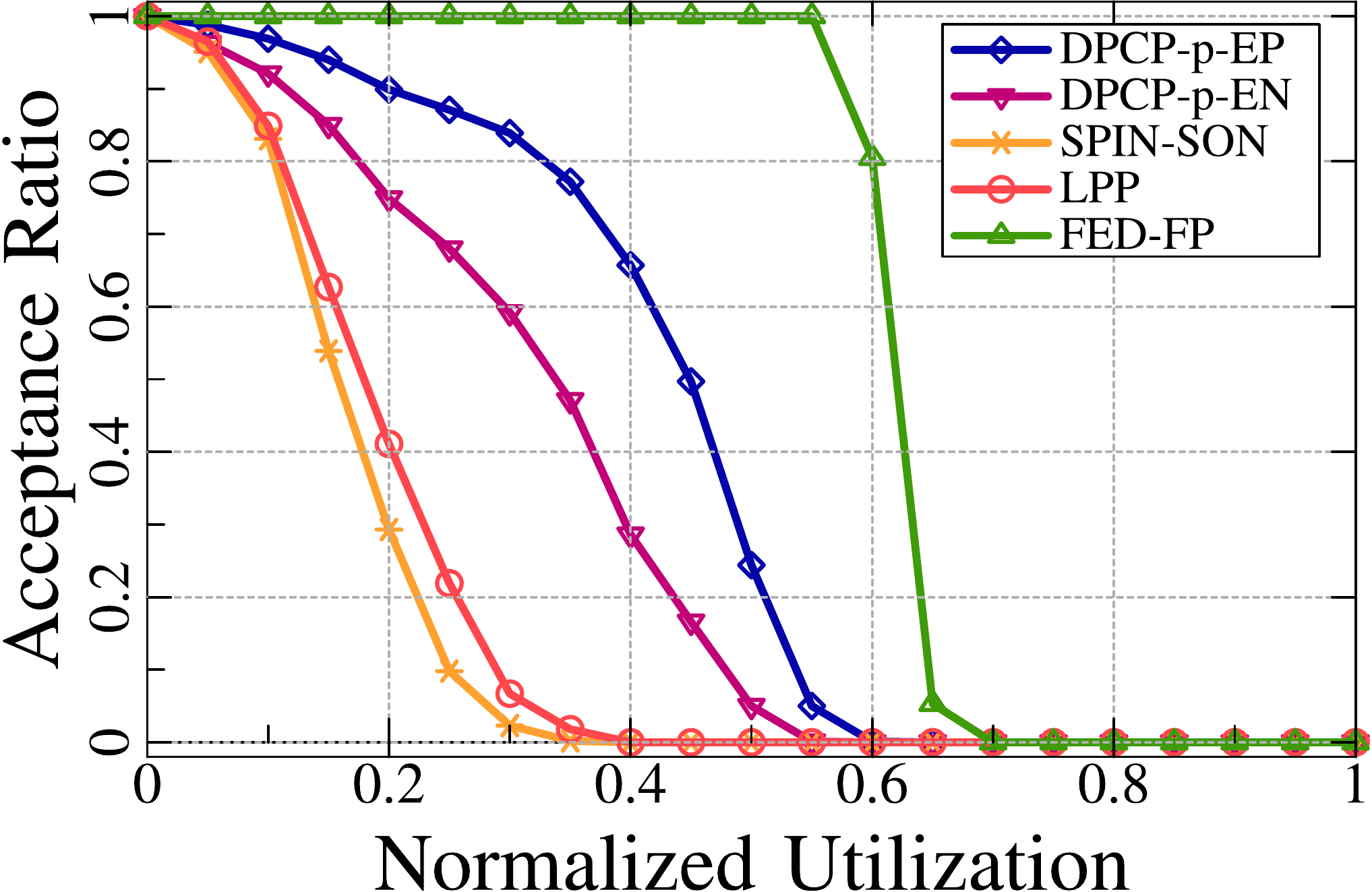}
\end{minipage}
\begin{minipage}{0.48\linewidth}
\vspace{1mm}
\centering
{\footnotesize (a) $U^{\text{avg}}=1.5$.}
\end{minipage}
\begin{minipage}{0.48\linewidth}
\vspace{1mm}
\centering
{\footnotesize (b) $U^{\text{avg}}=1.5$.}
\end{minipage}
\begin{minipage}{0.48\linewidth}
\centering
\includegraphics[width=4cm]{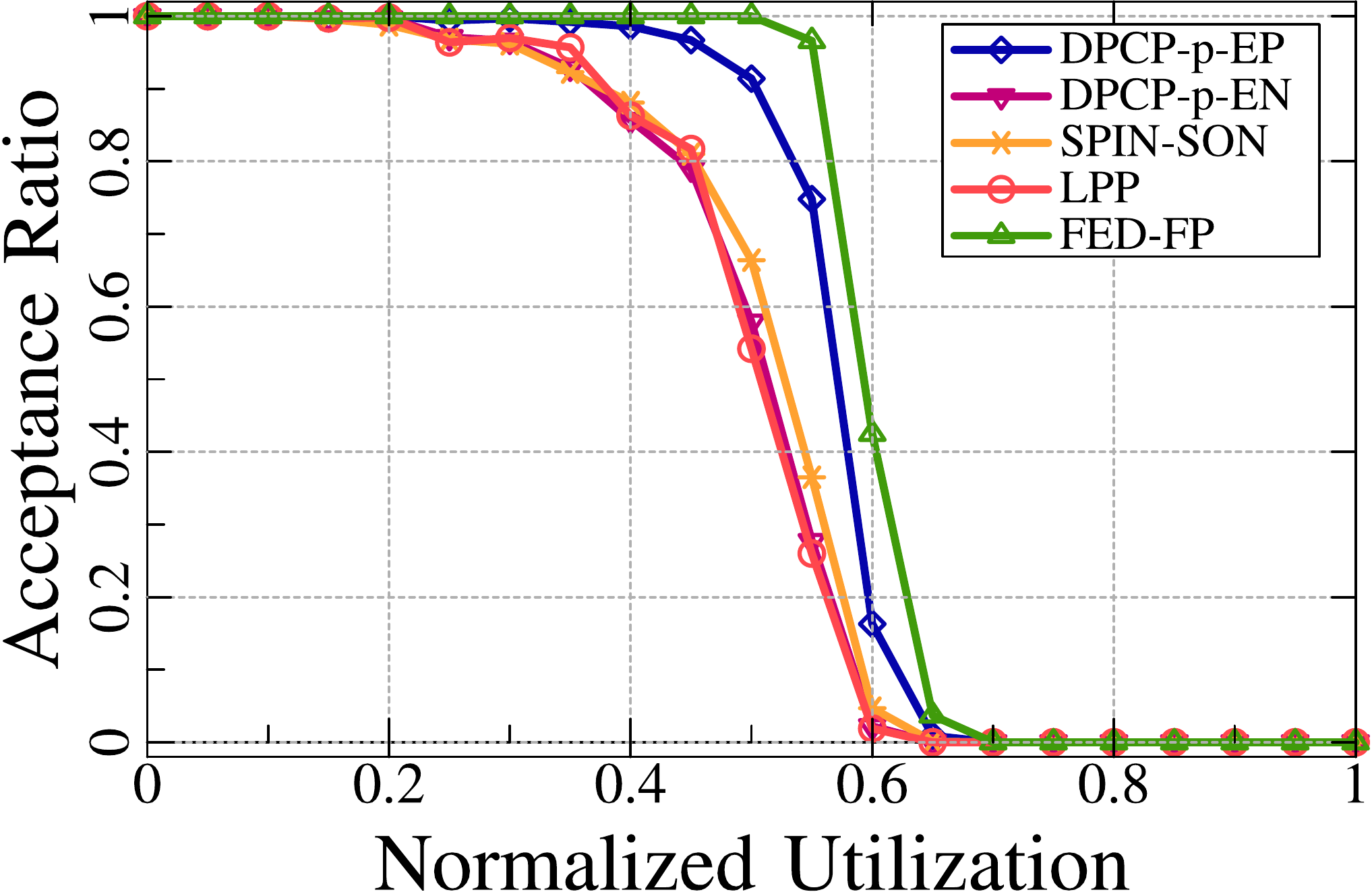}
\end{minipage}
\begin{minipage}{0.48\linewidth}
\centering
\includegraphics[width=4cm]{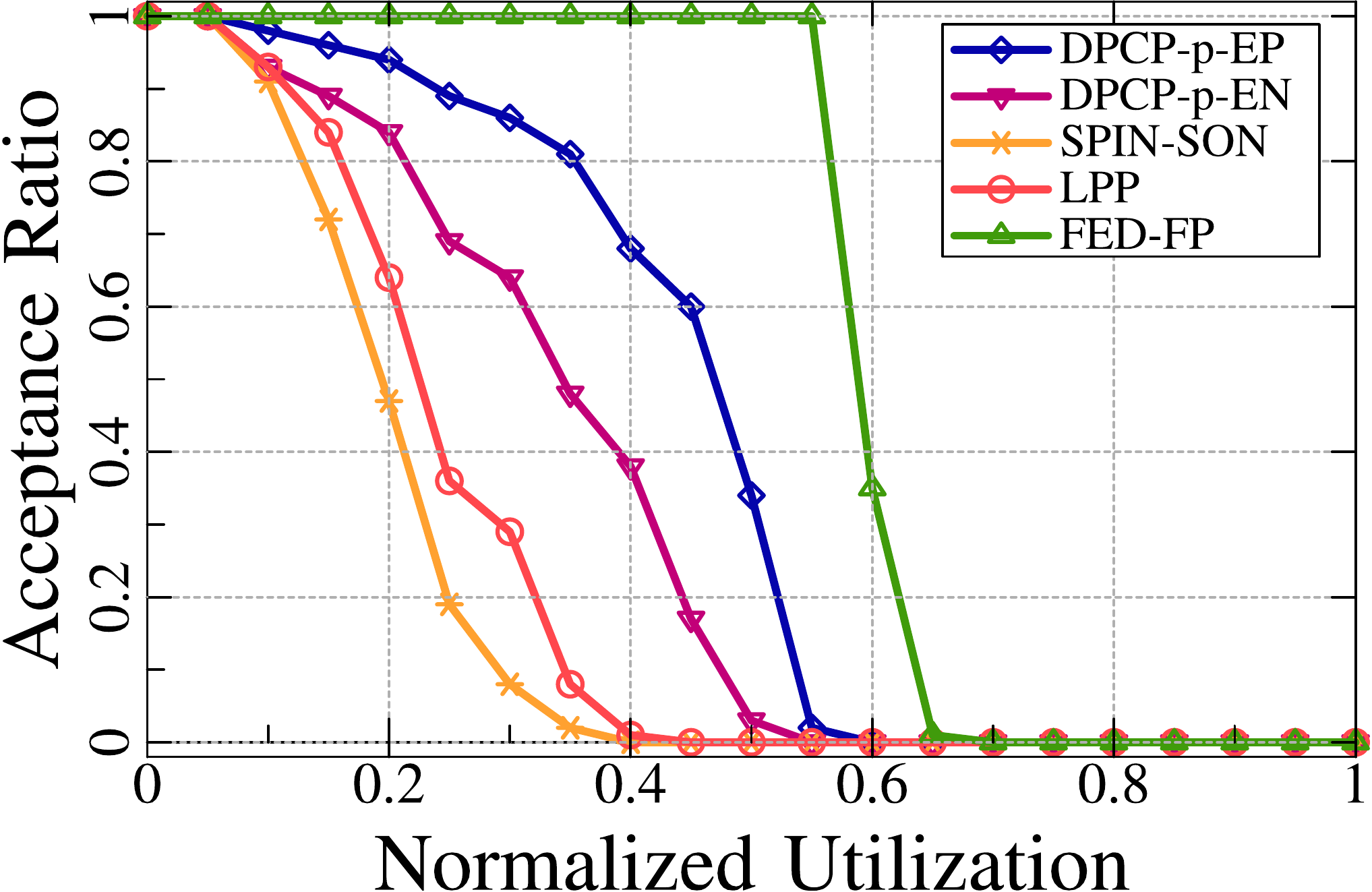}
\end{minipage}
\begin{minipage}{0.48\linewidth}
\vspace{1mm}
\centering
{\footnotesize (c) $U^{\text{avg}}=2$.}
\end{minipage}
\begin{minipage}{0.48\linewidth}
\vspace{1mm}
\centering
{\footnotesize (d) $U^{\text{avg}}=2$.}
\end{minipage}
\caption{Experiment results for $N_{i,q}\in [1,50]$, $L_{i,q}\in [50\mu s, 100\mu s]$, where $m=16$, $n_r\in [4,8]$, $p^r=0.5$ for (a) and (c), and $m=32$, $n_r\in [8,16]$, $p^r=1$ (b) and (d).}
\label{fig:results}
\end{figure}

\begin{center}
{\footnotesize{\bf Table 2.} Statistic for Dominance.}\\
\vspace{1mm}
\footnotesize{
\begin{tabular}{|c|c|c|c|c|}
\hline
& DPCP-p-EP & DPCP-p-EN & SPIN-SON & LPP \\
\hline
DPCP-p-EP & N/A & 216(100\%) & 215(99.5\%) & 216(100\%) \\
\hline
DPCP-p-EN & 0(0.0\%) & N/A & 104(48.1\%) & 87(40.3\%) \\
\hline
SPIN-SON & 0(0.0\%)& 10(4.6\%) & N/A & 39(18.1\%) \\
\hline
LPP & 0(0.0\%) & 32(14.8\%) & 38(17.6\%) & N/A \\
\hline
\end{tabular}}
\end{center}

\begin{center}
{\footnotesize{\bf Table 3.} Statistic for Outperformance.}\\
\vspace{1mm}
\footnotesize{
\begin{tabular}{|c|c|c|c|c|}
\hline
& DPCP-p-EP & DPCP-p-EN& SPIN-SON & LPP \\
\hline
DPCP-p-EP & N/A & 216(100\%) & 216(100\%) & 216(100\%) \\
\hline
DPCP-p-EN & 0(0.0\%)& N/A & 138(63.9\%) & 158(73.1\%) \\
\hline
SPIN-SON & 0(0.0\%) & 78(36.1\%) & N/A & 114(52.8\%) \\
\hline
LPP & 0(0.0\%)& 58(26.9\%) & 102(47.2\%) & N/A \\
\hline
\end{tabular}}
\end{center}

\section{Related Work}
\label{sec:related-work}
Real-time scheduling algorithms and analysis techniques for independent parallel tasks have been widely studied in the literature~\cite{DBLP:conf/rtns/FonsecaNN17,DBLP:conf/ecrts/LiCALGS14,
DBLP:conf/rtss/JiangLGW16,DBLP:journals/tc/MelaniBBMB17,DBLP:conf/ipps/BonifaciDM17,
DBLP:conf/rtss/CasiniBNB18a}, where shared resources are not modeled explicitly. %The design and analysis of real-time locking protocols for parallel tasks is still at the early stage~\cite{DBLP:journals/tpds/DinhLAGL18,DBLP:conf/dac/JiangGLY19}.

The study of multiprocessor real-time locking protocols stems from the DPCP~\cite{DBLP:conf/rtss/RajkumarSL88} and the Multiprocessor Priority Ceiling Protocol (MPCP)~\cite{DBLP:conf/icdcs/Rajkumar90}. Empirical studies~\cite{DBLP:conf/rtas/Brandenburg13} showed that the DPCP has better schedulability performance than the MPCP. Based on the DPCP, Hsiu et al.~\cite{DBLP:conf/emsoft/HsiuLK11} presented a dedicated-core scheduling. More recently, Huang et al.~\cite{DBLP:conf/rtss/HuangYC16} proposed the ROP scheduling. However, the work in \cite{DBLP:conf/rtss/RajkumarSL88,DBLP:conf/icdcs/Rajkumar90,DBLP:conf/rtas/Brandenburg13,
DBLP:conf/emsoft/HsiuLK11,DBLP:conf/rtss/HuangYC16} all assumes sequential task models. Although the locking protocols that are originally used for sequential tasks, e.g., \cite{DBLP:conf/rtss/RajkumarSL88,DBLP:conf/icdcs/Rajkumar90}, might be used to handle concurrent requests of parallel tasks, no work on the corresponding analysis has been established in the literature. In this paper, we extend the DPCP to support parallel real-time tasks and present the schedulability analysis.

%The design and analysis of real-time locking protocols for parallel tasks is still at the early stage.

Recently, there is significant progress on the scheduling of parallel real-time tasks, e.g., partitioned~\cite{DBLP:conf/rtss/CasiniBNB18a}, semi-partitioned~\cite{DBLP:conf/ipps/BonifaciDM17}, global~\cite{DBLP:conf/rtns/FonsecaNN17,DBLP:journals/tc/MelaniBBMB17}, federated~\cite{DBLP:conf/ecrts/LiCALGS14}, and decomposition-based scheduling~\cite{DBLP:conf/rtss/JiangLGW16}. However, no study on locking protocols for the state-of-the-art scheduling approaches other than the federated scheduling have been reported in the literature, so far as we know. For federated scheduling, Dinh et al.~\cite{DBLP:journals/tpds/DinhLAGL18} studied the schedulability analysis for spinlocks. Jiang et al.~\cite{DBLP:conf/dac/JiangGLY19} developed a semaphore protocol called LPP under federated scheduling. Both \cite{DBLP:journals/tpds/DinhLAGL18} and \cite{DBLP:conf/dac/JiangGLY19} assume local execution of resource requests. The presented DPCP-p is based on a distributed synchronization framework, where requests to global resources are conducted on designated processors. In this way, the contention on each resource can be isolated to the designated processor such that blocking among tasks can be better managed.

\iffalse
Notably, since locking requests are performed by each vertex, the locking protocols that are originally used for sequential tasks, e.g., \cite{DBLP:conf/rtss/RajkumarSL88,DBLP:conf/icdcs/Rajkumar90}, might be used. However, no work on the corresponding analysis has been established in the literature. Therefore, we focus on the performance of existing locking protocols for parallel real-time tasks under federated scheduling in the experiments.
\fi

\section{Conclusion}
\label{sec:conclusion}

This paper for the first time studies the distributed synchronization framework for parallel real-time tasks with shared resources. We extend the DPCP to DAG tasks for federated scheduling and develop analysis techniques and partitioning heuristic to bound the task response time. More precise blocking analysis based on the concrete DAG structure would be an interesting future work. 

\bibliographystyle{abbrv}
\bibliography{references}
\end{document}